\documentclass[preprint,12pt,authoryear]{elsarticle}

\usepackage{amssymb}
\usepackage{amsthm}
\usepackage{bm}
     \usepackage{tabularx} 
  \usepackage{amsmath} \numberwithin{equation}{section} 
 \allowdisplaybreaks[3]

\newcommand{\bI}{{\bf I}}
\newcommand{\bo}{{\bf 0}}
\newcommand{\bone}{{\bf 1}}

\newcommand{\bX}{{\bf X}}
\newcommand{\bx}{{\bf x}}
\newcommand{\bxi}{{\bm \xi}}
\newcommand{\bS}{{\bm \Sigma}}
\newcommand{\btea}{{\bm \theta}}
\newcommand{\bmu}{{\bm \mu}}
\newcommand{\mb}{\mbox}
\newcommand{\md}{\mbox{d}}
\newcommand{\mS}{\bf S}

\newtheorem{theorem}{Theorem}[section]
\newtheorem{lemma}{Lemma}[section]

\newtheorem{definition}{Definition}[section]
\newtheorem{corollary}{Corollary}[section]
\newtheorem{proposition}{Proposition}[section]

\usepackage{fancyhdr}
\journal{Journal}

\begin{document}
\begin{frontmatter}



\title{Modifications of 
Wald's score tests on large dimensional covariance matrices structure.}



\author[j]{Dandan Jiang\corref{cor1}\fnref{myfootnote}}
\fntext[myfootnote]{Supported by Project 11471140 from NSFC.}
\corref{Dan}
\cortext[cor1]{Corresponding author:Dandan Jiang}
\address[j]{School of Mathematics, \\
Jilin University, \\
2699  QianJin Street, \\
Changchun {\rm 130012}, China.}
\ead{jiangdandan@jlu.edu.cn}

\author[j]{Qibin Zhang}
\ead{jasonzhangqb@gmail.com}

\begin{abstract}

This paper considers testing the covariance matrices structure based on Wald's score test in large dimensional setting. The hypothesis $H_0: \bS  =\bS_0 $ for a given matrix  $\bS_0$, which covers the identity hypothesis test and sphericity hypothesis test as the special cases, is reviewed by the generalized CLT (Central Limit Theorem) for the linear spectral statistics of large dimensional sample covariance matrices from \cite{J15}.  The proposed tests can be applicable for large dimensional non-Gaussian variables in a wider range. Furthermore, the simulation study is provided to compare the proposed tests with other large dimensional covariance matrix tests for evaluation of their performances. As seen from the simulation results, our proposed tests are feasible for large dimensional data without restriction of population distribution and provide the accurate and steady empirical sizes, which are almost around the nominal size.

\end{abstract}


\end{frontmatter}


\section{Introduction} \label{Int}
The  hypothesis of testing   the covariance  matrices structure is represented as 
\begin{equation}
H_0 : \bS=\bS_0  \quad \mb{v.s.} \quad  H_1: \bS \neq \bS_0, \label{H1}
\end{equation}
which  plays an important role in classical multivariate
analysis, and it has been widely used in the social and behavioural sciences  etc.  However, the rapid  development of data 
 acquisition and computer science have  challenged the traditional statistical methods,
  because they are established on the basis of fixed dimension $p$ and fail to analyze  large dimensional data. This  inspired growing interests in  
 proposing new testing procedures designed for  the large dimensional data.
In this aspect, \cite{Johnstone}, \cite{LedoitWolf}  and \cite{Srivastava} gave the large dimensional tests  based on some different  functions of  the distance between the null and alternative hypotheses under the Gaussian assumption. Bai et al.(2009)   and Jiang  et al.(2012)  derived the asymptotic distributions of their proposed test statistics in large limiting scheme  $p/n \rightarrow c \in [0,1)$ and $p/n \rightarrow c \in [0,1]$ with $p<n$, respectively. 
Recently, Chen et al.(2010) proposed  a  nonparametric test    regardless of the limiting  behavior of $p/n$ and   Cai and Ma(2013)  presented an  optimal test on 
 high dimensional covariance matrix  from a minimax point of view. Also, Jiang(2015) used the RMT (Random Matrix Theory) to correct  the Rao's score test, which is applicable  for large dimensional data  without restriction of  Gaussian assumption.

  Motivated by the above works, we reviewed the testing  problem (\ref{H1})  and proposed some  new tests based on amending the Wald's score tests by RMT.  
  We first  derived Wald's score test statistic for the hypothesis  (\ref{H1})  and  refined it to satisfy  the  enhanced version of the CLT(Central Limit Theorem) for LSS ( Linear Spectral Statistics) of large dimensional sample covariance matrices in Jiang(2015), which is 
  the basic tools used in  the amendment process. Then  we  proposed  the correction to  Wald's score test statistics and derived their asymptotic distribution  under large  limiting scheme  $p/(n-1) \rightarrow q \in   [0,1)$, which provided more accurate and steady empirical sizes.

The remainder  of the article  is arranged as below. A quick survey  of  Wald's score test and  an enhanced version of the  large dimensional CLT cited from Jiang(2015)  are reviewed in Section \ref{Pre},  and the  details of   Wald's score tests on  covariance  structure are also derived in this part.
In Section \ref{New}, the  new  testing statistics  based on the aforementioned Wald's score tests are proposed by the large dimensional  limiting tools in RMT,  which compensate for  the effects of  large  dimensionality.    Then the simulations are conducted  to evaluate the performance
of  our proposed  tests compared with other large dimensional tests in Section \ref{Sim}. Finally,  conclusions and comments are drawn  
in the  Section \ref{Con}, and the proofs and derivations are listed in the \ref{App}.  


\section{ Preliminary}  \label{Pre}

Let $\chi=(\bx_1,\cdots,\bx_n)$  denote a sample from  a random vector $\bX$ following  the population distribution  $F_{\bX}(x,\btea)$, where  $\btea$ is an unknown parameter.
Make some notations as below:
\begin{itemize}
\item  $f_{\bX}(x,\btea)$ is the density function of $\bX$;
\item  $U(\bX,\btea)=\displaystyle\frac{\md}{\md \btea}\mb{ln} f_{\bX}(x,\btea)$ is the score vector  of $\bX$;
\item  $I (\bX,\btea)=\mb E(U(\bX,\btea)U'(\bX,\btea))$  is the  information matrix of $\bX$, which is  also calculated by Hessian matrix $H(\bX,\btea)$ as below:
\[I (\bX,\btea)=- \mb E (H(\bX,\btea))=- \mb E (\frac{\md^2}{\md \btea^2}\mb{ln} f_{\bX}(x,\btea))\]
\end{itemize}
Then the log-likelihood, the score function and the information matrix of the sample are given by 
$l(\chi,\btea)=\sum\limits_{i=1}^{n}\mb{ln} f(\bx_i,\btea)$, $U(\chi,\btea)=\sum\limits_{i=1}^{n}U(\bx_i,\btea)$ and $I (\chi,\btea)=nI (\bx_1,\btea)$, respectively. 
So the definition of  Wald's score test statistic is described as below:
\begin{definition}
For hypothesis  test $H_0 : \btea=\btea_0  $, Wald's score test statistic(\mb{WST}) is  defined as 
\[
\mb{WST}(\chi, \btea_0)=(T_{\btea}(\chi)-\btea_0)'\cdot I(\chi,T_{\btea}(\chi))\cdot(T_{\btea}(\chi)-\btea_0),
\]
 where $\btea_0=(\theta_{01},\cdots,\theta_{0p})'$ is a known vector,
 $\chi$ is a random sample from the population distribution, $T_{\btea}(\chi)$ is  the maximum likelihood estimator of the parameter $\btea$ and $I(\chi,T_{\btea}(\chi))$ is the information matrix substituting $\btea$ with $T_{\btea}(\chi)$. Then $ \mb{WST}(\chi, \btea_0) $ tends to a  $\chi^2_p$ limiting distribution, which is a $\chi^2$-distribution with $p$ degrees of freedom，
  as $n \rightarrow \infty$ under  $H_0$. (Wald,1943).
\end{definition}

To figure  out the Wald's score test statistic for the hypothesis test (\ref{H1}),   we suppose that  the sample $\chi=(\bx_1,\cdots,\bx_n)$ follows a normal distribution  $N_p(\bmu, \bS)$.  Denote 
$\btea=(\bmu', \mb{vec}'(\bS))' $, where $\mb{vec} ( \cdot )$ is the vectorization  operator. Then  the logarithm  of the density of the sample $\chi$ is
\[l(\chi,\btea)=-\frac{np}{2}\mb{ln}(2\pi)-\frac{n}{2}\ln|\bS|-\frac{1}{2}\sum\limits_{i=1}^{n}\mb{tr}\left(\bS^{-1}(\bx_i-\bmu)(\bx_i-\bmu)'\right).\] 
By some derivations we arrive at  the following several  results [for details  see the Appendix A.1 in \cite{J15}], 
\begin{equation}
U(\chi, \btea)=:\left(
\begin{array}{c}
U_1 (\chi, \btea) \\
U_2(\chi, \btea)
\end{array}
\right)=\left(
\begin{array}{c}
n\bS^{-1}(\hat\bmu-\bmu)  \\
\displaystyle\frac{n}{2}\mb{vec}(\bS^{-1}(\mS\bS^{-1}-\bI_p))
\end{array}
\right) \label{scorevec}
\end{equation}
where  $\frac{\md}{\md \btea}=
\left(
\begin{array}{cc}
\frac{\md}{\md \bmu}   \\
\frac{\md}{\md \mb{vec}(\bS)}
\end{array}
\right)
$ is a $p(p+1) \times 1$ vector and  
\[\hat\bmu=\displaystyle\frac{1}{n}\sum\limits_{i=1}^n\bx_i,   \quad \mS=\displaystyle\frac{1}{n}\sum\limits_{i=1}^{n}(\bx_i-\bmu)(\bx_i-\bmu)'.\] 
According to the definition of Hessian matrix, 
$H(\chi, \btea)=\frac{\md^2}{\md \btea^2}l(\chi,\btea)=: 
\left(
\begin{array}{ccc}
 H_{11} & H_{12}    \\
 H_{21} & H_{22}   
\end{array}
\right)
$,  where $H_{22}$  is the part for the parameter $\bS$ and written as 
\begin{eqnarray}
H_{22}&=&\frac{n}{2}\frac{\md \mb{vec}(\bS^{-1}(\mS\bS^{-1}-\bI_p))}{\md \mb{vec}'(\bS)}\label{H22d}\\ \nonumber 
&=&\frac{n}{2}\frac{\md \mb{vec}(\bS^{-1})}{\md \mb{vec}'(\bS)} \frac{\md \mb{vec}(\bS^{-1}\mS\bS^{-1}-\bS^{-1})}{\md  \mb{vec}'(\bS^{-1})}\\\nonumber
&=& -\frac{n}{2} (\bS^{-1} \otimes \bS^{-1} )(\mS\bS^{-1} \otimes \bI_p+\bI_p\otimes \mS\bS^{-1}-\bI_{p^2}). 
\end{eqnarray}
where $\otimes$ is the Kronecker products.
The information matrix is denoted as   $I(\chi, \btea)=: 
\left(
\begin{array}{ccc}
 I_{11}(\chi, \btea) & I_{12} (\chi, \btea)   \\
 I_{21}(\chi, \btea) & I_{22} (\chi, \btea)  
\end{array}
\right)
$,  then we get
   the part  of   information matrix for $\bS$ 
\[I_{22}(\chi, \btea)=-\mb{E}\left[H_{22}(\chi, \btea)\right]= \frac{n}{2} (\bS^{-1} \otimes \bS^{-1} )\]
 by the expectation   $\mb{E} (\mS)=\mb E [(\bX-\bmu)(\bX-\bmu)']=\bS$.

As stated in  \citet{Gombay}, if there are no restrictions on $\bmu$,  the parameter $\bmu$ and $\mS$ in the 
score vector are  replaced by its maximum likelihood estimator $\hat\bmu$ and $\hat\bS$, where 
\begin{equation}
\hat\bS=\displaystyle\frac{1}{n}\sum\limits_{i=1}^{n}(\bx_i-\hat\bmu)(\bx_i-\hat\bmu)'. \label{Sigmahat}
\end{equation}  
Then  we have $U_1 (\chi, \btea)=\bo$,  so only $I_{22}(\chi, \btea)$  is involved in the calculation of the 
Wald's score test statistic. Therefore, we have

\begin{proposition}
Wald's score test statistic  for testing $H_0 : \bS=\bS_0 $ with no constrains on $\bmu$ has the 
following form 
\begin{equation}
\mb{WST}(\chi, \bS_0)=\frac{n}{2}\mb{tr}[(\mb{\bI}_p-\bS_0\widehat{\bS}^{-1})^2] \label{WST1}
\end{equation}
where $\chi=(\bx_1,\cdots,\bx_n)$ is a sample from $N_p(\bmu,\bS)$, and the test statistic 
$\mb{WST}(\chi, \bS_0)$ tends to a $\chi^2$-distribution  with   $\displaystyle\frac{p(p+1)}{2}$  degrees of freedom   under $H_0$ when $n \rightarrow \infty$.
\label{Prop1}
\end{proposition}

\begin{proof}
Because the parameter $\bmu$ and $\mS$ in the 
score vector are  replaced by  $\hat\bmu$ and $\hat\bS$, If there is no constrain on $\bmu$. That means $U_1 (\chi, \btea)=\bo$, consequently  $I(\chi, \btea)=: 
\left(
\begin{array}{ccc}
 \bo& \bo   \\
\bo &\displaystyle \frac{n}{2} (\hat\bS^{-1} \otimes \hat\bS^{-1} )
\end{array}
\right)
$.  For a further step,  $T_{\btea}(\chi)=(\hat\bmu', \mb{vec}'(\hat\bS))'$ is   the maximum likelihood estimator of the parameter $\btea$,  so  we have 
\begin{eqnarray*}
\mb{WST}(\chi, \bS_0)
&=&\frac{n}{2}\mb{vec}'(\widehat\bS-\bS_0) 
 (\widehat\bS^{-1} \otimes \widehat\bS^{-1} )
\mb{vec}(\widehat\bS-\bS_0) \\
&=&\frac{n}{2}\mb{vec}'(\widehat\bS-\bS_0) 
\mb{vec}(\widehat\bS^{-1}(\widehat\bS-\bS_0)\widehat\bS^{-1}) \\
&=&\frac{n}{2}\mb{tr}((\widehat\bS-\bS_0) 
\widehat\bS^{-1}(\widehat\bS-\bS_0)\widehat\bS^{-1})\\
&=&\frac{n}{2}\mb{tr}[(\mb{\bI}_p-\bS_0\widehat{\bS}^{-1})^2] 
\end{eqnarray*}
\end{proof}

\begin{corollary}
Wald's score test statistic  for testing $H_0 : \bS=\bI_p $ with no constrains on $\bmu$ has the 
following form 
$$
\mb{WST}(\chi, \bI_p)=\frac{n}{2}\mb{tr}[(\mb{\bI}_p-\widehat{\bS}^{-1})^2]
$$
where $\chi=(\bx_1,\cdots,\bx_n)$ is a sample from $N_p(\bmu,\bS)$, and the test statistic 
$\mb{WST}(\chi, \bI_p)$ tends to a $\chi^2$-distribution  with   $\displaystyle\frac{p(p+1)}{2}$  degrees of freedom under $H_0$ when $n \rightarrow \infty$.
\end{corollary}

\begin{corollary}
Wald's score test statistic  for testing $H_0 : \bS=\gamma\bI_p $ with no constrains on $\bmu$ has the 
following form 
$$
\mb{WST}(\chi, \gamma\bI_p)=\frac{n}{2}\mb{tr}[(\mb{\bI}_p-\frac{\mb{tr}(\widehat\bS)}{p}\widehat{\bS}^{-1})^2]
$$
where $\chi=(\bx_1,\cdots,\bx_n)$ is a sample from $N_p(\bmu,\bS)$ and $\gamma>0$ is an unknown parameter. 
Then the test statistic 
$\mb{WST}(\chi, \gamma\bI_p)$ tends to a $\chi^2$-distribution  with   $\displaystyle\frac{p(p+1)}{2}-1$  degrees of freedom under $H_0$ when $n \rightarrow \infty$.
\end{corollary}
\begin{proof}
Replace the $\bS_0$ by $\hat\gamma\bI_p$ according to (\ref{WST1}), where $\hat\gamma=\displaystyle\frac{\mb{tr}(\widehat\bS)}{p}$ is the maximum likelihood estimator of $\gamma$.
\end{proof}

\subsection{CLT for LSS of a high-dimensional sample covariance matrix }

Following  the above Proposition~\ref{Prop1}, the statistics of Wald's score test  for the hypothesis (\ref{H1}) can be transformed into the trace of a matrix concerned with the sample covariance matrix, i.e. a function of the eigenvalues of  the sample covariance matrix. It  exactly meets the requirements in the CLT for LSS of large dimensional covariance matrices in \cite{BS04}, so we can modify the classical Wald's score tests by this limiting tool. In order to expand the usable range of the modified Wald's score tests,    a quick survey of an enhanced  version of the CLT for LSS of large dimensional covariance matrices is cited from  \cite{J15}, which excludes the  strict condition on the 4th moment for  a wider usage.  Before quoting, we first introduce some basic concepts and notations.

Suppose  
$(\bxi_1, \cdots,\bxi_{n}) $  to be 
an $\mbox{i.i.d}$ sample from some $p$-dimensional distribution with mean $\bo_p$ and
covariance matrix $\bI_p$, where $\bxi_i = (\xi_{1i}, \xi_{2i},\cdots , \xi_{pi})'$. 
The corresponding sample covariance matrix is
\begin{equation}
\mS_n={1\over{n}}\sum\limits_{i=1}^{n}\bxi_i\bxi_i^*.\label{Sn}
\end{equation}
where $(\cdot)^*$ is  conjugate transpose.
For simplicity, $F^q, F^{q_n}$  denote the Mar\v{c}enko-Pastur law of index
 $q~ \mbox{and}  ~ q_n$ respectively, where $q_n=\frac{p}n \rightarrow q \in [0, +\infty)$.
 $F_n^{\mS_n}$ represents the  ESD (Empirical Spectral Distribution) of the matrix $\mS_n$.
Define the empirical process
$G_n : = \{G_n(f)\}$ indexed by $\mathcal{A}$ ,
\begin{equation}
G_n(f)= p\cdot \int_{-\infty}^{+\infty} f(x)\left[F^{\mS_n}_n-
F^{q_n}\right] (\md x), ~~~~~~\quad f \in  \mathcal{A},\label{Gdef}
\end{equation}
 where $\mathcal{U}$ is an open set of the complex plane
  including  the surppoting  set of $F^q $
and $\mathcal{A}$   be the set of analytic functions $f :
\mathcal{U} \mapsto \mathbb{C}.$  Define
\[\kappa=
\left\{
\begin{array}{cc}
 2, & \text{if the~}  \bxi- \text{variables are real,\quad\quad} \\
  1,& \text{if the~}  \bxi- \text{variables are complex.} 
\end{array}
\right.
\]
Then an enhanced  version of the CLT for LSS of large dimensional covariance matrices from    \citet{J15} (Lemma 2.1) is quoted as following:

\begin{lemma}

Assume: \\
   $ f_1, \cdots ,f_k \in \mathcal{A}$,  $\{\xi_{ij}\}$  are  $i.i.d.$
random variables, such that  $\mb{E}\xi_{11}=0,~ \mb{E}{|\xi_{11}|^2}=\kappa-1,
~\mb{E}{|\xi_{11}|}^4 < \infty$  and the $\{\xi_{ij}\}$ satisfy the condition
\[\frac{1}{np}\sum\limits_{ij} \mb{E}|\xi_{ij}|^4I(|\xi_{ij}|\geq \sqrt{n}\eta) \rightarrow 0\]
for any fixed $\eta>0$.  Moreover,
 $\displaystyle\frac{p}{n}=q_n \rightarrow q
\in [0, +\infty) $
as  $n, p \rightarrow \infty$  and $E(\xi_{11}^4)=\beta+\kappa+1,$ where $\beta$ is a  constant.\\[1mm]
Then the random vector $\left(G_n(f_1), \cdots 
,  G_n(f_k) \right)$  forms a tight sequence by index $n$, and  
it  weakly converges to a $k$-dimensional Gaussian
vector with mean vector
\begin{eqnarray}
\mu(f_j)&=&-\frac{\kappa-1}{2\pi i} \oint f_j(z) \frac{q\underline{m}^3(z)(1+\underline{m}(z))}{[(1-q)\underline{m}^2(z)+2\underline{m}(z)+1]^2} \md z \label{04mean1}\\
&&-\frac{\beta q }{2 \pi i} \oint f_j(z)\frac{\underline{m}^3(z)}{(1+\underline{m}(z))[(1-q)\underline{m}^2(z)+2\underline{m}(z)+1]} \md z,
\label{04mean2}
\end{eqnarray}
and covariance
function
\begin{eqnarray}
\upsilon\left(f_j,
f_\ell\right)&=&-\frac{\kappa}{4\pi^2}\oint\oint\frac{f_j(z_1)f_\ell(z_2)}{(\underline{m}(z_1)-\underline{m}(z_2))^2}
\md\underline{m}(z_1)\md \underline{m}(z_2)  \label{04var1}\\
&&-\frac{\beta q}{4\pi^2}\oint\oint\frac{f_j(z_1)f_\ell(z_2)}{(1+\underline{m}(z_1))^2(1+\underline{m}(z_2))^2}
\md\underline{m}(z_1)\md \underline{m}(z_2), \label{04var2}
\end{eqnarray}
where  $ j,\ell \in \{1, \cdots,
k\}$, and  $\underline{m}(z)\equiv m_{\underline{F}^q}(z)$ is the
Stieltjes Transform of ~ $\underline{F}^q\equiv (1-q)I_{[0,
    \infty)}+qF^q$. The contours all contain the support of $F^q$ and 
   non overlapping  in both (\ref{04var1})  and    (\ref{04var2}).\\
\label{CLT}
\end{lemma}


 \section{ The Proposed Testing Statistics } \label{New}

 In this section,   $\chi=(\bx_1,\cdots,\bx_n)$ is set as  an independent and identically distributed  sample from 
a $p$ dimensional random vector $\bX$ with mean $\bmu$ and covariance matrix $\bS$. To test on  the structure of  covariance  matrix, we consider the hypothesis  
\[
H_0 : \bS=\bS_0  \quad \mb{v.s.} \quad  H_1: \bS \neq \bS_0, 
\]
which covers  the identity hypothesis test  $H_0 : \bS=\bI_p $ and  sphericity  hypothesis test  $H_0 : \bS=\gamma\bI_p $ as  the special cases. 

 It has been  well studied under the normal distribution assumption with the classical setting of fixed $p$, such as  Anderson (2003), Nagao (1973) and John (1971) etc.  Also,   Wald's score test was given in   Wald(1943). They all  lost their effectiveness  as the dimension $p$ was much higher  and performed  even worse for the non-Gaussian variables.  So we 
 hope to correct the Wald's score test for the hypothesis (\ref{H1}) by using Lemma~\ref{CLT}, which makes the correction  applicable for large dimensional data and non-Gaussian assumption.

  Set 
$\widetilde\bxi_i=\bS_0^{-\frac{1}{2}}(\textbf{x}_i-\bmu),$
then the array $\{\widetilde\bxi_i\}_{i=1, \cdots, n}$ contains  $p$-dimensional
standardized variables under $H_0$.  If the parameter $\bmu$ is unknown,  the sample mean is used instead. Then the Lemma \ref{CLT} should be applied by $n-1$  instead of $n$  by \cite{ZhengBaiYao}. Therefore, we define $ \widetilde\bS=\displaystyle\frac{n}{n-1}\widehat\bS\bS_0^{-1}$, then $\widetilde\bS$ has  the same LSD with 
$\mS_{n-1}$ defined in (\ref{Sn}) with $n$ substituted by $n-1$.  
Let 
\begin{equation}
\widetilde{\mb{WST}}(\chi, \bS_0)=\frac{n}{2}\mb{tr}[(\mb{\bI}_p-\widetilde\bS^{-1})^2], \label{WSTt}
\end{equation}
  it is also natural  to apply the  Lemma \ref{CLT} with  $n-1$  instead of $n$ to the modified  Wald's score test statistic $\displaystyle \frac{2}{n}\widetilde{\mb{WST}}(\chi, \bS_0)$. 
Thus,  the theorem of large dimensional  Wald's score test is proposed  as below:  
\begin{theorem}
Suppose that the conditions of Lemma \ref{CLT} hold, for hypothesis test $H_0 :  \bS  =\bS_0$,  $\widetilde{\mb{WST}}(\chi, \bS_0)$ is defined as (\ref{WSTt}),  set $p/(n-1)=q_n \rightarrow q \in [0, 1)$, 
and  $f(x)= (1-\displaystyle \frac{1}{x})^2$.  Then,  under $H_0$ and when $n\rightarrow\infty$, the correction to Wald's score test statistics is 
\begin{equation}
  CWST(\chi, \bS_0)=\upsilon(f)^{-\frac{1}{2}}\left[ \frac{2}{n}\widetilde{\mb{WST}}(\chi, \bS_0)-p \cdot
    F^{q_n}(f)- \mu(f)\right] \Rightarrow N \left( 0,
  1\right),
  \label{CWST1}
\end{equation}
where $F^{q_n}$ is the Mar\v{c}enko-Pastur law of  index  $ q_n$,  and $ F^{q_n}(f),\mu(f)$ and $\upsilon(f)$ are calculated in (\ref{limitWST}), (\ref{meanWST}) and (\ref{varWST}),  respectively.
  \label{CWSTth}
  \end{theorem}
\begin{proof}
By  the  derivation  (\ref{WSTt}), we have 
\begin{eqnarray*}
 \frac{2}{n}\widetilde{\mb{WST}}(\chi, \bS_0)&=&\mb{tr}[(\mb{\bI}_p-\widetilde\bS^{-1})^2] \\
  &=& \sum\limits_{i=1}^{p} \left(1-\frac{1}{\lambda_i^{\widetilde\bS}}\right)^2= p \cdot \int (1-\frac1{x})^2
  \md F^{\widetilde\bS}_n(x)\\
  &=&p \cdot \int f(x) \md\left(F^{\widetilde\bS}_n(x)-F^{q_n}(x)\right) +p \cdot
  F^{q_n}(f),
\end{eqnarray*}
where  $(\lambda_i^{\widetilde\bS}),  i=1,\cdots, p$ and $F_n^{\widetilde\bS}$   are the eigenvalues and the   ESD of the matrix $\widetilde\bS$, respectively.
   $F^{q_n} (f)$ denotes the integral of the function $f(x) $ by the density corresponding to the Mar\v{c}enko-Pastur law of index $q_n$, that is 
\begin{equation}
F^{q_n}(f)=\int_{-\infty}^{\infty}f(x) \md F^{q_n}(x) = 1-\displaystyle\frac{2}{(1-q_n)}+\displaystyle\frac{1}{(1-q_n)^3}, \quad \text{if }  \quad  0 \leq q_n<1,\label{limitWST}
\end{equation}
which is calculated
in the  \ref{App}.
%

As the definition in (\ref{Gdef}), we have 
\begin{equation}
 G_n(f)= p \cdot \int f(x) d\left(F^{\widetilde\bS}_n(x)-F^{q_n}(x)\right)=\frac{2}{n}\widetilde{\mb{WST}}(\chi, \bS_0)-p \cdot F^{q_n}(f).\label{WSTESD-pLSD}
\end{equation}
By Lemma \ref{CLT}, $G_n(f)$  weakly converges  to a Gaussian vector  with the mean
\begin{equation}
\mu(f)=\frac{-(\kappa-1)q(2q^2-5q-1)}{(1-q)^4}+\frac{\beta q(2q^2-3q-1)}{(q-1)^3} \label{meanWST}
\end{equation}
and variance
\begin{equation}
\upsilon(f)=\frac{2\kappa q^2(2q^3-12q^2+18q+1)}{(q-1)^8}+\frac{4 \beta q^3(2-q)^2}{(q-1)^6}.\label{varWST}
\end{equation}
which are calculated in  the \ref{App}.    Then, by Lemma~\ref{CLT} and  (\ref{WSTESD-pLSD}),  we arrive at
\[
\frac{2}{n}\widetilde{\mb{WST}}(\chi, \bS_0)-p\cdot F^{q_n}(f)~\Rightarrow~ N\left(\mu(f),
\upsilon(f)\right),\label{WSTconnec}
\]
  Finally, we get
\begin{eqnarray}
 CWST(\chi, \bS_0)=\upsilon(f)^{-\frac{1}{2}}\left[ \frac{2}{n}\widetilde{\mb{WST}}(\chi, \bS_0)-p \cdot
    F^{q_n}(f)- \mu(f)\right] \Rightarrow N \left( 0,
  1\right)\nonumber
\end{eqnarray}
\end{proof}

For  the identity and  sphericity  hypothesis tests,  the corollaries are as below  :
\begin{corollary}
For testing $H_0 : \bS=\bI_p $ with no constrains on $\bmu$,  the conclusion of Theorem \ref{CWSTth} still holds,
only with the test statistic  $\widetilde{\mb{WST}}(\chi, \bI_p)$ in (\ref{CWST1}) is revised by
\begin{eqnarray}
\widetilde{\mb{WST}}(\chi, \bI_p)=\frac{n}{2}\mb{tr}[(\mb{\bI}_p-(\frac{n}{n-1}\widehat{\bS})^{-1})^2] \label{CWST2}.
\end{eqnarray}
\end{corollary}

\begin{corollary}
For testing $H_0 : \bS=\gamma\bI_p $ with no constrains on $\bmu$,   the conclusion of Theorem \ref{CWSTth} still holds,
only with the test statistic  $\widetilde{\mb{WST}}(\chi, \bI_p)$ in (\ref{CWST1}) is revised by

$$
\widetilde{\mb{WST}}(\chi, \gamma\bI_p)=\frac{n}{2}\mb{tr}[(\mb{\bI}_p-\hat\gamma(\frac{n}{n-1}\widehat{\bS})^{-1})^2]
$$
where  $\hat\gamma=\displaystyle\frac{\mb{tr}(\frac{n}{n-1}\widehat\bS)}{p}$ is the maximum likelihood estimator of $\gamma$.
\end{corollary}

 
\section{Simulation Study}\label{Sim}

 Without loss of the generality, the  identity hypothesis test  $H_0 : \bS=\bI_p $  is investigated in this section. The simulations  are  conducted to compare our proposed corrections to Wald's score test (CWST)  with other large dimensional tests on covariance matrices, like the tests in \cite{LedoitWolf} (LWT),    \cite{Cai}  (CMT) and classical  covariance tests in  \cite{Nagao} (NHT) and \cite{Wald} (WST).  We generate $\mb{i.i.d}$ random samples $\chi=(\bx_1,\cdots,\bx_n)$ from two scenarios of  the  $p$-dimensional populations   under the null hypothesis:
 \begin{description}
\item[Normal Assumption:] Following a $p$-dimensional  normal  distribution with mean $\mu_0\bone_p$ and covariance matrix $\bI_p$, where $\mu_0=2$ and $\bone_p$ denotes a vector with that all elements  are 1.  
\item  [Gamma Assumption:]   Following a $p$-dimensional  Gamma  distribution with all the components are $\mb{i.i.d.}$  from the distribution of  Gamma (4,0.5).
\end{description}

For each  scenario, The empirical sizes and powers are reported  with 10,000 replications at $\alpha=0.05$ significance  level. 
We chose the cases from $n=300,  p=80,120,160,200$ and $n=500,  p=80,160,240,320$   and the mean parameter is supposed to be unknown and substituted by the sample mean during the calculations.

For the  alternative hypothesis,  the population covariance matrix is  designed as the tridiagonal matrix $\bS=(\sigma_{ij})$ for different population assumptions, where for $\rho \in (0,1)$  and
\begin{equation}
\sigma_{ij}=\left\{
\begin{array}{ll}
 1, &  i=j     \\
  \rho, &   |i-j|=1   \\
  0, &    |i-j|>1   
\end{array}
\right.\label{tridiag}
\end{equation}

 Simulation results of empirical sizes  and powers  of four tests  are listed in the Table  \ref{tab:1}, which  includes  our proposed  CWST,  the test in \cite{LedoitWolf} (LWT) and \cite{Nagao} (NHT) and Wald's score test (WST). The comparison of  empirical sizes between our proposed  CWST and   the test in \cite{Cai} (CMT)  is also presented in Table~\ref{tab:2}.

%
  
 \begin{table}[htdp]
\caption{Empirical sizes and powers of the comparative tests for  $H_0 : \bS=\bI_p$  at $\alpha=0.05$ significance level  for Normal and Gamma Assumptions with 10,000  replications.  The alternative hypothesis is    the tridiagonal matrix  $\bS=(\sigma_{ij})$  with $\sigma_{ij}$ defined in (\ref{tridiag}).  \label{tab:1}}
\begin{center}
\begin{tabularx}{13.5cm}{X|XXXXXX}
\hline
& \multicolumn{6}{c}{ $\tiny{(n,p)=(300,80)}$ }  \\ \hline
&Sizes &\multicolumn{2}{c}{ Powers } &Sizes &\multicolumn{2}{c}{ Powers } \\
&$\rho=0$& $\rho=0.05$& $\rho=0.15$& $\rho=0$& $\rho=0.05$& $\rho=0.15$\\
&\multicolumn{3}{c}{ Normal }  &\multicolumn{3}{c}{ Gamma}\\
CWST  & 0.0649   & 0.1303    & 0.9861    &0.0640     & 0.1273    & 0.9748  \\
LWT     & 0.0704    & 0.2199   &   1           &0.2290      &0.4698     &  1 \\
NHT     &  0.1191   &  0.2632   &   0.9998  &0.2934     &0.4701     & 0.9998  \\
WST    &  1            &  1            &   1           & 1             & 1             & 1  \\\hline
& \multicolumn{6}{c}{ $\tiny{(n,p)=(300,160)}$ }  \\ \hline
&Sizes &\multicolumn{2}{c}{ Powers } &Sizes &\multicolumn{2}{c}{ Powers } \\
&$\rho=0$& $\rho=0.05$& $\rho=0.18$& $\rho=0$& $\rho=0.05$& $\rho=0.18$\\
&\multicolumn{3}{c}{ Normal }  &\multicolumn{3}{c}{ Gamma}\\
CWST & 0.0682     &0.1071     &0.9785     & 0.0688    &0.1033     & 0.9746  \\
LWT &0.0874         &0.2619     & 1             & 0.2607    & 0.5358    &   1\\
NHT & 0.1779        &0.3213    &  1           &  0.3486   & 0.5182    &  1 \\
WST& 1                 &1             &  1             & 1             &   1          & 1  \\\hline
& \multicolumn{6}{c}{ $\tiny{(n,p)=(500,160)}$ }  \\ \hline
&Sizes &\multicolumn{2}{c}{ Powers } &Sizes &\multicolumn{2}{c}{ Powers } \\
&$\rho=0$& $\rho=0.05$& $\rho=0.12$& $\rho=0$& $\rho=0.05$& $\rho=0.12$\\
&\multicolumn{3}{c}{ Normal }  &\multicolumn{3}{c}{ Gamma}\\
CWST  & 0.0615   & 0.1713    &0.9820              &0.0596     & 0.1574    & 0.9725  \\
LWT     & 0.0694    & 0.3982   &   1           &0.2279      &0.6774     &  1 \\
NHT     &  0.1276   &  0.4152   &   1          &0.3067     &0.6321    & 1  \\
WST    &  1            &  1            &   1           & 1             & 1             & 1  \\\hline
& \multicolumn{6}{c}{ $\tiny{(n,p)=(500,320)}$ }  \\ \hline
&Sizes &\multicolumn{2}{c}{ Powers } &Sizes &\multicolumn{2}{c}{ Powers } \\
&$\rho=0$& $\rho=0.05$& $\rho=0.15$& $\rho=0$& $\rho=0.05$& $\rho=0.15$\\
&\multicolumn{3}{c}{ Normal }  &\multicolumn{3}{c}{ Gamma}\\
CWST & 0.0627     &0.1139     &0.9349     & 0.0640   &0.1127     & 0.9304  \\
LWT &0.0912         &0.4669     & 1             & 0.2887    & 0.7368    &   1\\
NHT & 0.1904        &0.4750    &   1            &  0.3857   & 0.6341    &  1 \\
WST& 1                 &1             &  1             & 1             &   1          & 1  \\\hline
\end{tabularx}
\end{center}
\end{table}

Note from the Table~\ref{tab:1},  it is  easily to find out the informations as below:

\begin{description}
\item[ (i) ]  The traditional  Wald's score test  (WST) is completely unworkable  for large dimensional  data in respect of their  empirical  sizes,  which all equals to 1 for any case listed in the table.
\item[ (ii) ]  The  empirical  sizes of the  test in  Nagao(1973) (NHT) also deviate far from the nominal test size $5\%$,   and they increase with the  dimension $p$, 
 especially  worse for Gamma Assumption. 
\item[ (iii) ]  For the cases of  the relatively smaller  dimensions in the table, like $p=80$  with $n=300$ or $p=160$  with $n=500$, the test in Ledoit and wolf (2002) (LWT)  behaves well and provides higher 
powers. However,  the empirical  sizes of LWT rise up against the nominal level as $p$ increases, and it shows a even worse result  under the Gamma Assumption,  where our proposed CWST is  still active.
\item[ (iv) ] The empirical sizes of our proposed test CWST
 are almost around the nominal size 5\%  for both distribution assumptions.   Although the empirical powers of
  our proposed CWST are not as higher as others' for the  very small $\rho$,     the other tests work well in one place but fail in another because of their rising empirical sizes.  Besides,，
   the empirical powers of
  the proposed CWST quickly increase to 1 with a slight  upward adjustment  of $\rho$.
 
\end{description}

Furthermore, it  was  showed in \cite{Cai}  that  CMT  provided the optimal powers, which  uniformly dominated that of the corrected LRTs by random matrix theory  over the asymptotic  regime.
 However, in contrast to CMT, the proposed CWST  performs  more steady empirical sizes,  which is supported by the following brief table~\ref{tab:2}. The corresponding power comparisons
 for both scenarios are similar to the ones  between  LWT and CWST under the Normal Assumption.  It shows a relatively slow ascent of the powers by our proposed test  due to the involvement of  the inverse of the sample covariance matrix. But the powers of our test  will quickly rise to 1  if the null and alternative hypotheses are not much close.   And if it will bring  serious consequences  when 
the null hypothesis  occurred but not found in the practice,  we should be strict  to Type I errors. Therefore, it is better to choose our proposed test in such a situation, which provides precise empirical sizes.
\begin{table}[htdp]
\caption{Empirical sizes of the  tests proposed  in  this paper(CWST) and  \cite{Cai}  (CMT)  for  $H_0 : \bS=\bI_p$  at $\alpha=0.05$ significance level.
The simulations are repeated  10,000 times   for Normal and Gamma Assumptions respectively.  \label{tab:2}}
\begin{center}
\begin{tabular}{l|cccl}
\hline

~~$n=300$~~&CWST &CMT &CWST &CMT \\ \hline
 &\multicolumn{2}{c}{ Normal } &\multicolumn{2}{c}{ Gamma } \\
~~$p=80$   & 0.0649   & 0.0710   & 0.0640    & 0.0772~~        \\
~~$p=120$ &0.0652    &0.0789    & 0.0687    & 0.0795       \\
~~$p=160$ &0.0682    &0.0872    & 0.0688    & 0.0905        \\
~~$p=200$ &0.0719    &0.0947    & 0.0713    & 0.0973        \\ \hline
~~$n=500$~~&CWST &CMT &CWST &CMT \\  \hline
 &\multicolumn{2}{c}{ Normal } &\multicolumn{2}{c}{ Gamma } \\ 
~~$p=80$   &0.0587    & 0.0620   &0.0584     & 0.0603        \\
~~$p=160$ & 0.0615   & 0.0699   & 0.0596    & 0.0678        \\
~~$p=240$ & 0.0623   & 0.0855   & 0.0603    & 0.0832        \\
~~$p=320$ & 0.0627   & 0.0914   & 0.0640    & 0.0953        \\  \hline
\end{tabular}
\end{center}
\label{default}
\end{table}%


\section{Conclusion}\label{Con}

In this paper, the  new tests for the covariance matrices structure  based on modification of Wald's score test are proposed by RMT.  They are feasible
for large dimensional data  without restriction of  population distribution and provide the accurate empirical sizes.  
However, it must be noted that the proposed CWST cannot be used for $p\geq n$, even if the case of $p/n \rightarrow 1$, since it is involved  with the inverse of the sample covariance matrix. That's also 
the reason why it gives   a relatively slow ascent of the powers.   This problem  is  brought by the  statistic  on which our correction is based, rather than the the idea of correction or 
the large dimensional CLT  in random matrix theory we used.  So  it can be expected that the similar idea of methods can  be used  for  other suitable statistics to avoid these disadvantages. For example,
Jiang(2015) 
proposed a   testing statistic for the large dimensional covariance structure test  based on amending Rao's score tests, which  is applicable for the case of $p>n$ and Non-Gaussian assumption. In the future,
we may also look forward to the applications of random matrix theory in more statistical inferences.
 
\section*{Acknowledgement}
The author thanks the reviewers for their helpful comments and suggestions to make an improvement of this article. This research was supported by the National Natural Science Foundation of China 11471140.





\appendix
\section{Proofs}\label{App}


%
\begin{itemize}

\vskip 0.2in
\item {\bf  Calculation of  $F^{q_n} (f)$ in  (\ref{limitWST}).}
\vskip 0.2in

Because 
\[
F^{q_n}(f)=\int_{-\infty}^{\infty}f(x) \mb(d) F^{q_n}(x)=\int_{-\infty}^{\infty}\frac{(1-\displaystyle\frac{1}{x})^2}{2\pi x q_n}\sqrt{(b_n-x)(x-a_n)}\md x
\]
Use the substitution  $x=1+q_n-2\sqrt{q_n}\cos \theta=-2\sqrt{q_n}(\cos\theta+d_0)$, where $0 \leq \theta \leq \pi $ and  $d_0=-\displaystyle\frac{1+q_n}{2\sqrt{q_n}}$ is a constant.
  Then 
\begin{eqnarray*}
  &&F^{q_n}(f)=\int^{b_n}_{a_n}\frac {(1-\displaystyle\frac{1}{x})^2}{2\pi
    xq_n}\sqrt{(b_n-x)(x-a_n)} \md x\\
  &&=\frac{2}{\pi
    }\int_0^{\pi}\frac{(q_n-2\sqrt{q_n}\cos\theta)^2}
    {(1+q_n-2\sqrt{q_n}\cos\theta)^3}\sqrt{\sin^2\theta}\sin\theta \md\theta\\
  &&=\frac{1}{\pi
    }\int_0^{2\pi}\frac{[-2\sqrt{q_n}(\cos\theta+d_0)-1]^2\sin^2\theta }
    {\left[-2\sqrt{q_n}(\cos\theta+d_0)\right]^3}\md\theta\\
    &&=\frac{1}{\pi
    }\int_0^{2\pi}\left[\frac{\sin^2\theta}{-2\sqrt{q_n}(\cos\theta+d_0)}-\frac{2\sin^2\theta}{4q(\cos\theta+d_0)^2}-\frac{\sin^2\theta  }{8q_n^{\frac{3}{2}}(\cos\theta+d_0)^3}\right]\md\theta\\
\end{eqnarray*}

First, we have 
\begin{eqnarray}
\int_{0}^{2\pi}\frac{1}{\cos\theta+d_0}\md \theta&=&\int_{0}^{2\pi}\frac{1}{\cos^2\frac{\theta}{2}-\sin^2\frac{\theta}{2}+d_0}\md \theta\nonumber\\
&=&\int_{0}^{2\pi}\frac{2\md \frac{\theta}{2}}{(1-\tan^2\frac{\theta}{2}+d_0\sec^2\frac{\theta}{2})\cos^2\frac{\theta}{2}}\nonumber\\
&=&\int_{0}^{2\pi}\frac{2\md \tan\frac{\theta}{2}}{d_0+1+(d_0-1)\tan^2\frac{\theta}{2}}\nonumber\\
&=&\frac{2}{\sqrt{d_0^2-1}}\displaystyle\int_{0}^{2\pi}
\frac{1}{1+\left(\sqrt{\frac{d_0-1}{d_0+1}}\tan\frac{\theta}{2}\right)^2}
\md \left(\sqrt{\frac{d_0-1}{d_0+1}}\tan\frac{\theta}{2}\right)\nonumber\\
&=&\frac{-2\pi}{\sqrt{d_0^2-1}}\label{eqA1chucosd}
\end{eqnarray}
So the first part of the integral, 
\begin{eqnarray}
&&\int_0^{2\pi}-\frac{\sin^2\theta}{2\sqrt{q_n}(\cos\theta+d_0)}\md\theta \nonumber\\
&=&-\frac{1}{2\sqrt{q_n}}\int_0^{2\pi}\frac{1-\cos^2\theta  }{\cos\theta+d_0}\md\theta\nonumber\\
&=&-\frac{1}{2\sqrt{q_n}}\int_0^{2\pi} (-\cos\theta+d_0+\frac{1-d_0^2}{\cos\theta+d_0})\md\theta\nonumber\\
&=&\frac{\pi}{2q_n}(1+q_n-|1-q_n|)\label{eqA1sin2chucosd}
\end{eqnarray}

Secondly,
the middle  part of the integral can be calculated by the equation eq.(\ref{eqA1chucosd}), which is 
\begin{eqnarray}
&&\int_0^{2\pi}-\frac{2\sin^2\theta}{4q_n(\cos\theta+d_0)^2}\md\theta=\frac{1}{2q_n}\int_0^{2\pi}\left(\frac{-1}{\cos\theta+d_0}\right)'\sin \theta\md \theta\nonumber\\
&&~~~~~~=\frac{1}{2q_n}\int_0^{2\pi}\frac{\cos\theta}{\cos\theta+d_0}\md \theta=\frac{\pi}{q_n}\left(1-\frac{1+q_n}{|1-q_n|}\right).\label{eqA2sin2chucosd2}
\end{eqnarray}
For the third part of the limiting integral  $F^{q_n}(f)$, the following integral is calculated first.
\begin{eqnarray*}
&&\int_{0}^{2\pi}\frac{1}{(\cos\theta+d_0)^2}\md \theta\nonumber\\
&=&\int_{0}^{2\pi}\frac{\sin^2\theta}{(\cos\theta+d_0)^2}\md \theta 
+\int_{0}^{2\pi}\frac{\cos\theta-d_0}{\cos\theta+d_0}\md \theta
+\int_{0}^{2\pi}\frac{d_0^2}{(\cos\theta+d_0)^2}\md \theta\nonumber
\end{eqnarray*}
Then, by the eq. (\ref{eqA1chucosd}) and (\ref{eqA2sin2chucosd2}), we arrive at 
\begin{eqnarray}
&&\int_{0}^{2\pi}\frac{1}{(\cos\theta+d_0)^2}\md \theta\nonumber\\
&=&\frac{1}{d_0^2-1}\left[-\int_{0}^{2\pi}\frac{\sin^2\theta}{(\cos\theta+d_0)^2}\md \theta
-\int_{0}^{2\pi}\frac{\cos\theta-d_0}{\cos\theta+d_0}\md \theta\right]\nonumber\\
&=&\frac{8\pi q_n(1+q_n)}{|1-q_n|^3}\label{eqA1chucosd2}
\end{eqnarray}
So the third part of the limiting integral  $F^{q_n}(f)$ is 
\begin{eqnarray*}
&&\int_0^{2\pi}-\frac{\sin^2\theta}{8q_n^{\frac{3}{2}}(\cos\theta+d_0)^3}\md\theta\\
&=&\frac{1}{16q_n^{\frac{3}{2}}}\int_0^{2\pi}\left[\frac{-1}{(\cos\theta+d_0)^2}\right]' \sin \theta\md\theta
=\frac{1}{16q_n^{\frac{3}{2}}}\int_0^{2\pi}\frac{\cos \theta}{(\cos\theta+d_0)^2}\md\theta\\
&=&\frac{1}{16q_n^{\frac{3}{2}}}\left[\int_0^{2\pi}\frac{1}{\cos\theta+d_0}\md\theta-d_0\int_0^{2\pi}\frac{1}{(\cos\theta+d_0)^2}\md\theta\right]
=\frac{\pi}{|1-q_n|^3}
\end{eqnarray*}

Above all, $F^{q_n}(f)$ is presented below
\begin{eqnarray*}
F^{q_n}(f)&=& \frac{1}{\pi}\left[\frac{\pi}{2q_n}(1+q_n-|1-q_n|)+\frac{\pi}{q_n}\left(1-\frac{1+q_n}{|1-q_n|}\right)
+\frac{\pi}{|1-q_n|^3}\right]\\
&=&
\left\{
\begin{array}{cc}
  1-\displaystyle\frac{2}{(1-q_n)}+\displaystyle\frac{1}{(1-q_n)^3},& \quad\quad  ~~\mb{if} ~ 0 \leq q_n<1 ,  
  \\[4mm]
  \displaystyle\frac{1}{q_n}-\displaystyle\frac{2}{q_n(q_n-1)}+\displaystyle\frac{1}{(q_n-1)^3},& ~~\mb{if} ~ q_n>1.  
\end{array}
\right.
\end{eqnarray*}

It is worth to note that the $F^{q_n}(f)$
should plus the term $(1-1/x)^2\cdot (1-\displaystyle\frac{1}{q_n} ) $ at the origin $x=0$
if $q_n>1$, which is obviously infinity. That's  one of reasons that the correction to Wald's score test can only be used as $q_n<1$.
Then we arrive at 
\[F^{q_n}(f)= 1-\displaystyle\frac{2}{(1-q_n)}+\displaystyle\frac{1}{(1-q_n)^3}, \quad \text{if }  \quad  0 \leq q_n<1.\]


\vskip 0.2in
\item   {\bf Calculation of  $\mu (f)$ in (\ref{meanWST}).}
\vskip 0.2in

In the same way, with $H(t)=\mb{I}_{[1,\infty}(t)$, the first part of $\mu(f)$  is also obtained  by (9.12.13) in \cite{bookBS10}, 
$$
\mu_1(f)=(\kappa-1)\cdot\left(\frac{f\left(a(q)\right)+f\left(b(q)\right)}{4} -
\frac{1}{2\pi}\int_{a(q)}^{b(q)}\frac{f(x)}{\sqrt{4q-(x-1-q)^2}}\md x\right)
$$
where $a(q)=(1-\sqrt{q})^2$ and  $b(q)=(1+\sqrt{q})^2$.
For $f(x)= (1-\displaystyle\frac{1}{x})^2$, make a substitution
$x=1+q-2\sqrt{q}\cos\theta,~ 0\leq \theta \leq \pi$, then  

\begin{eqnarray*}
\mu_1(f)&=&(\kappa-1)\left(\frac{f\left(a(q)\right)+f\left(b(q)\right)}{4} -\frac{1}{4\pi}\int_{0}^{2\pi}
f(1+q-2\sqrt{q}\cos\theta)\md \theta\right)\\
&=&(\kappa-1)\bigg(\frac{q^4-6q^3+9q^2+4q}{2(1-q)^4}\\
&&\quad\quad\quad-\frac{1}{4\pi}\int_0^{2\pi}
\left[1-\frac{2}{-2\sqrt{q}(\cos\theta+d^*_0)}+\frac{1}{4q(\cos\theta+d^*_0)^2}\right] \md \theta\bigg)\\
&=&(\kappa-1)\left(\frac{q^4-6q^3+9q^2+4q}{2(1-q)^4}-\frac{1}{2}+\frac{1}{|1-q|}-\frac{1+q}{2|1-q|^3}\right)
\end{eqnarray*}
where $d^*_0=-\displaystyle\frac{1+q}{2\sqrt{q}}$ is an  analogy to the constant $d_0$ with $q$ instead of $q_n$, and the calculation can also be based on on (\ref{eqA1chucosd}) and (\ref{eqA1chucosd2}). Because the correction to Wald's score test can only be applied to the case $q_n<1$, so we only choose the first case
\[\mu_1(f)=\frac{-(\kappa-1)q(2q^2-5q-1)}{(1-q)^4},\quad \mb{if} ~ 0 \leq q<1.
 \]

For the second part of $\mu(f)$,  by  (\ref{04mean2})  we have
\[ 
\mu_2(f)=- \frac{\beta q }{2 \pi i} \oint (1-\frac{1}{z})^2\frac{\underline{m}^3(z)}{(1+\underline{m}(z))[(1-q)\underline{m}^2(z)+2\underline{m}(z)+1]} \md z,
\]
Recall the equation (9.12.12)   in \cite{bookBS10}
\begin{equation}
 z=-\frac{1}{\underline{m}(z)}+\frac{q}{1+\underline{m}(z)}, \label{zzmm}
\end{equation}
 it is easily obtained that 
\[(1-\frac{1}{z})^2=\frac{[m^2-(q-2)m+1]^2}{[(q-1)m-1]^2}  \quad \quad \md z = \frac{(1-q)m^2+2m+1}{m^2(1+m)^2} \md m\]
where  $\underline{m}(z)$  is denoted as $m$  for simplicity if no confusion. Then we have 
\[ 
\mu_2(f)= -\frac{\beta q }{2 \pi i (q-1)^2} \oint \frac{m[m^2-(q-2)m+1]^2}{(m-\frac{1}{q-1})^2(1+m)^3} \md m,
\]
and the contour for the integral  of $m$  is obtained by solving the equation (\ref{zzmm}),  which  should  enclose the interval 
$
\left[-\frac{1}{1-\sqrt{q}},   -\frac{1}{1+\sqrt{q}}\right]
$ when $0\leq q< 1$.
Therefore,  -1 and $\frac{1}{q-1}$ are  the residues if $q<1$,   and the integral is calculated as 
\[\mu_2(f)=\frac{\beta q(2q^2-3q-1)}{(q-1)^3}. \]

 Finally, we  obtained 
 \[\mu(f)=\frac{-(\kappa-1)q(2q^2-5q-1)}{(1-q)^4}+\frac{\beta q(2q^2-3q-1)}{(q-1)^3}.\]  

 \vskip 0.2in
\item  {\bf  Calculation of  $\upsilon (f)$.}
 \vskip 0.2in

By Lemma~\ref{CLT},  we have 
\begin{eqnarray*}
\upsilon\left(f_j,
f_\ell\right)&=&-\frac{\kappa}{4\pi^2}\oint\oint\frac{f_j(z_1)f_\ell(z_2)}{(\underline{m}(z_1)-\underline{m}(z_2))^2}
\md\underline{m}(z_1)\md \underline{m}(z_2)  \\
&&-\frac{\beta q}{4\pi^2}\oint\oint\frac{f_j(z_1)f_\ell(z_2)}{(1+\underline{m}(z_1))^2(1+\underline{m}(z_2))^2}
\md\underline{m}(z_1)\md \underline{m}(z_2), 
\end{eqnarray*}
 and
\begin{eqnarray*}
f(z_1)f(z_2)&=&(1-\frac{1}{z_1})^2(1-\frac{1}{z_2})^2\\
&=&1-\frac2{z_1}-\frac2{z_2}+\frac1{z_1^2}+\frac1{z_2^2}+\frac4{z_1z_2}
-\frac2{z^2_1z_2}-\frac2{z_1z_2^2}+\frac1{z^2_1z^2_2}.
\end{eqnarray*} 

It is known that  $\upsilon(\textbf{1},\textbf{1})=0$,  where  \textbf{1} denote constant function which equals to 1.
For $z \in \mathbb{C}^{+}$, the following equation is given in \cite{bookBS10},
 \begin{equation}
 z=-\frac{1}{\underline{m}(z)}+\frac{q}{1+\underline{m}(z)}.
 \end{equation}\\
Still use  $m_i$ to simplify  $\underline{m}(z_i),~~i=1,2.$   For fixed $m_2$,  there is a contour enclosed $(q-1)^{-1}$ and -1 as  poles when $0\leq q<1$. 
 Then, for the first item of    $\upsilon (f)$, we have 
 \begin{eqnarray*}
   &\quad&\displaystyle {\oint \displaystyle\frac{1}{z_1}\cdot \frac{1}
     {(m_1-m_2)^2}\md m_1}\\
   &=&\displaystyle{\oint \frac{m_1(1+m_{1})}{(q-1)(m_1-\frac{1}{q-1})(m_1-m_2)^2} \md m_{1}}\\
      &=&\displaystyle{2\pi i\cdot  \frac{q}{(q-1)^3(m_2-\frac{1}{q-1})^{2}}}.
 \end{eqnarray*}
  and 
  \begin{eqnarray}
  &&\displaystyle {\oint \frac{1}{z_1^2}\cdot\frac{1}
     {(m_1-m_2)^2}\md m_1}\nonumber\\
     &=&\displaystyle{\oint \frac{m_1^2(1+m_{1})^2}{(q-1)^2(m_1-\frac{1}{q-1})^2}  (m_2-\frac{1}{q-1})^{-2}\left(1-\frac{m_1-\frac{1}{q-1}}{m_2-\frac{1}{q-1}}\right)^{-2}\md m_{1}}\nonumber\\
     &=& 4\pi i \left[\frac{q(1+q)}{(q-1)^5(m_2-\frac{1}{q-1})^2}+\frac{q^2}{(q-1)^6(m_2-\frac{1}{q-1})^3} \right]. \label{eqvar1chuz2}
\end{eqnarray}
So $\upsilon(\displaystyle\frac{1}{z_1^2}-\displaystyle\frac{2}{z_1}, ~\textbf{1})=0$.
Similarly,
$\upsilon(\textbf{1}, ~ \displaystyle\frac{1}{z_1^2}-\displaystyle\frac{2}{z_1})=0$.\\
Therefore, there are only four parts left, i.e. $\frac4{z_1z_2}
-\frac2{z^2_1z_2}-\frac2{z_1z_2^2}+\frac1{z^2_1z^2_2}$.   
Further,
\begin{eqnarray*}
\upsilon(\frac{1}{z_1},\frac{1}{z_2})
&=&-\frac{\kappa}{4\pi^2}\oint \frac{1}{z_2}\oint \frac{1}{z_1} \frac1{(m_1-m_2)^2}
\md  m_1 \md m_2 \\
&=&\displaystyle{\frac {\kappa q}{2\pi i (q-1)^3}\oint \frac{m_2(1+m_2)}{(q-1)(m_2-\frac{1}{q-1})^3}
\md m_2}\\
&=&\frac{\kappa q}{(q-1)^4}\end{eqnarray*}

\begin{eqnarray*}
\upsilon(\frac{1}{z_1},\frac{1}{z^2_2})&=&-\frac{\kappa}{4\pi^2}\oint \frac{1}{z^2_2}\oint \frac{1}{z_1} \frac1{(m_1-m_2)^2}
\md  m_1 \md m_2\\
     &=&\displaystyle{\frac {\kappa q}{2\pi
i (q-1)^5}\oint \frac{m^2_2(1+m_2)^2}{(m_2-\frac{1}{q-1})^4}
\md m_2}\\
&=& \frac{2 \kappa q(1+q)}{(q-1)^6}.
\end{eqnarray*}
Similarly, $
\upsilon(\frac{1}{z^2_1},\frac{1}{z_2})
= \displaystyle \frac{2 \kappa q(1+q)}{(q-1)^6}$.
For the last part $\upsilon(\displaystyle\frac{1}{z^2_1},\displaystyle\frac{1}{z^2_2})$, the integral is calculated  by eq.(\ref{eqvar1chuz2}) as below.
\begin{eqnarray*}
&&\upsilon(\frac{1}{z^2_1},\frac{1}{z^2_2})\\
&&=-\frac{\kappa}{4\pi^2}\oint \frac{1}{z^2_2}\oint \frac{1}{z^2_1} \frac1{(m_1-m_2)^2}
\md  m_1 \md m_2\\
   &&=\frac{\kappa}{\pi i}\oint \frac{m_2^2(1+m_{2})^2}{(q-1)^2(m_2-\frac{1}{q-1})^2} \left[\frac{q(1+q)}{(q-1)^5(m_2-\frac{1}{q-1})^2}+\frac{q^2}{(q-1)^6(m_2-\frac{1}{q-1})^3} \right]\md m_2 \\
   &&= \frac{2\kappa q(1+2q)(q+2)}{(q-1)^8}
\end{eqnarray*}
So the first item of    $\upsilon (f)$ is 
\begin{eqnarray*}
\upsilon_1(f)&=&4\upsilon(\frac{1}{z_1},\frac{1}{z_2})-2\upsilon(\frac{1}{z^2_1},\frac{1}{z_2})-2\upsilon(\frac{1}{z_1},\frac{1}{z^2_2})+\upsilon(\frac{1}{z^2_1},\frac{1}{z^2_2})\\
&=&\frac{2\kappa q^2(2q^3-12q^2+18q+1)}{(q-1)^8}\end{eqnarray*}
when $0 \leq q<1$.

Secondly, the latter item of    $\upsilon (f)$ is 
\[
\upsilon_2(f)=-\frac{\beta q}{4\pi^2}\oint\oint\frac{f(z_1)f(z_2)}{(1+m_1)^2(1+m_2)^2}
\md m_1\md m_2.
\]
Furthermore,
\begin{eqnarray*}
\oint\frac{f(z_1)}{(1+m_1)^2}\md m_1
&=&\oint \frac{[m_1^2-(q-2)m_1+1]^2}{[(q-1)m_1-1]^2(1+m_1)^2} \md m_1\\
&=&4\pi i \frac{q(2-q)}{(q-1)^3}
\end{eqnarray*}
since the  contour contains 
$\frac{1}{q-1}$ and  -1 as a residues if  $0\leq q \leq 1$.
Thus we get
\[
\upsilon_2(f)=-\frac{\beta q}{4\pi^2} \cdot \left(4\pi i \frac{q(2-q)}{(q-1)^3} \right)^2=\frac{4 \beta q^3(2-q)^2}{(q-1)^6}\]

 Finally, we  obtained 
 \[\upsilon(f)=\frac{2\kappa q^2(2q^3-12q^2+18q+1)}{(q-1)^8}+\frac{4 \beta q^3(2-q)^2}{(q-1)^6}.\]  
\end{itemize}




\end{document}